\documentclass[journal,twoside,web]{ieeecolor}
\usepackage{generic}
\usepackage[noadjust]{cite}

\usepackage{amssymb}
\usepackage{mathtools}
\usepackage{tikz}
\usepackage{pgfplots}
\usepackage{hyperref}

\usetikzlibrary{arrows,positioning,shapes,calc,fit}

\usepackage{amsfonts}
\usepackage{dsfont}
\usepackage{amsmath}
\usepackage{amssymb}
\usepackage{bm}
\usepackage{subcaption}
\usepackage{comment}

\newtheorem{theorem}{Theorem}
\newtheorem{definition}{Definition}
\newtheorem{proposition}{Proposition}

\newtheorem{lemma}{Lemma}
\newtheorem{corollary}{Corollary}

\newtheorem{remark}{Remark}

\newcommand{\mc}{\mathcal}

\newcommand{\Vcal}{\mc{V}}
\newcommand{\Gcal}{\mc{G}}
\newcommand{\Ecal}{\mc{E}}

\newcommand{\discretset}[2]{ \left\{#1, \dots, #2 \right\}}

\newcommand{\Vset}{\discretset{1}{N}}

\newcommand{\sign}[1]{\text{sgn}\left(#1\right)}

\newcommand{\R}{\mathbb{R}}

\newcommand{\one}{\mathbf{1}}

\def\BibTeX{{\rm B\kern-.05em{\sc i\kern-.025em b}\kern-.08em
    T\kern-.1667em\lower.7ex\hbox{E}\kern-.125emX}}
\begin{document}

\title{Global synchronization of multi-agent systems with nonlinear interactions}

\author{Anthony Couthures$^1$
, Vineeth S. Varma$^{1,2}$
, Samson Lasaulce$^{1,3}$
, Irinel-Constantin Mor\u{a}rescu$^{1,2}$
\thanks{Manuscript submitted March 11$^{th}$ 2025.  (\emph{Corresponding author Anthony Couthures})}
\thanks{*This work has been funded by the CNRS MITI project BLESS and by the project DECIDE funded under the PNRR I8 scheme by the Romanian Ministry of Research.}
\thanks{$^1$associated with Universit\'e de Lorraine, CNRS, CRAN, F-54000 Nancy, France. {\tt\small anthony.couthures@univ-lorraine.fr}}%
\thanks{$^2$associated with Automation Department, Technical University of Cluj-Napoca, Memorandumului 28, 400114 Cluj-Napoca, Romania.}%
\thanks{$^3$associated with Khalifa University, Abu Dhabi, UAE.}
}

\maketitle

\begin{abstract}
The paper addresses the synchronization of multi-agent systems with continuous-time dynamics interacting through a very general class of monotonic continuous signal functions that covers estimation biases, approximation of discrete quantization, or state-dependent estimation. Our analysis reveals that, in the setup under consideration, synchronization equilibria are exactly the fixed points of the signal function. We also derive intuitive stability conditions based on whether the signal underestimates or overestimates the state of the agents around these fixed points. Moreover, we show that network topology plays a crucial role in asymptotic synchronization. These results provide interesting insights into the interplay between communication nonlinearity and network connectivity, paving the way for advanced coordination strategies in complex systems.
\end{abstract}

\begin{IEEEkeywords}
	Consensus dynamics, Multi-agent systems, Nonlinear interactions, Synchronization
\end{IEEEkeywords}

\section{INTRODUCTION}

Information exchange is a fundamental aspect of both multi-agent systems (MAS) and networked control systems (NCS), as the agents' dynamics are shaped by the information received from their neighbors. In many existing works on MAS, the interaction is considered to be perfect, i.e., the agents have access to the exact state of their neighbors  \cite{degrootReachingConsensus1974}. This assumption simplifies the analysis but does not hold in real-world scenarios where communication is constrained due to bandwidth limitations, quantization effects, or noise. To address this limitation, several works \cite{Martins2008,frasca2012continuous,chowdhuryContinuousOpinionsDiscrete2016,ceragioliConsensusDisagreementRole2018b,prisantDisagreementLimitCycles2024,couthuresAnalysisOpinionDynamics2024} have investigated quantized communication models in which agents exchange only discrete or coarse information about their states. 

On the other hand, nonlinear consensus dynamics has also gained significant attention, particularly in models where agents interact through smooth and odd signals \cite{grayMultiagentDecisionMakingDynamics2018,fontanRoleFrustrationCollective2022,bizyaevaNonlinearOpinionDynamics2023}. These approaches aim to capture more realistic scenarios where interactions are governed by nonlinear feedback mechanisms, often inspired by biological \cite{grayMultiagentDecisionMakingDynamics2018} or social \cite{bizyaevaNonlinearOpinionDynamics2023} systems. However, these works primarily analyze behaviors around the neutral synchronization equilibrium using bifurcation theory. 
While this approach provides valuable theoretical insights, it does not fully capture the global behavior of the system, particularly in scenarios where multiple stable synchronization equilibria or complex nonlinear interactions arise. This limitation highlights the need for an analysis that goes beyond local bifurcations around the origin.

In this paper, we present a framework for nonlinear consensus dynamics in which the interaction between agents is modeled by a general class of signals. By considering arbitrary monotonic, potentially non-smooth, Lipschitz-continuous signal functions, our analysis extends existing studies while preserving analytical tractability. This general setting encompasses classical linear consensus as well as continuous and piecewise approximations of discontinuous quantization, providing new insights on how communication nonlinearity influences the stability and convergence properties of the synchronization manifold in multi-agent systems.

The key contributions of this paper are as follows: 
\begin{enumerate}
\item We analyze a nonlinear consensus model in which agent interactions are described by monotonic Lipschitz-continuous signal functions. This formulation not only includes the traditional linear consensus model as a particular case but also accommodates a wide range of nonlinear continuous functions, as well as piecewise approximations of discontinuous quantization functions. 
\item We characterize the equilibria of the system and show that synchronization equilibria correspond to the fixed points of the signal function. For that, we introduce intuitive notions of over- and under- estimation functions to analyze whether the synchronization occurs. Based on these concepts, we establish conditions for the stability of synchronized equilibria. We also highlight how the shape of the signal function influences the convergence behavior. Moreover, we provide conditions under which solutions converge to a locally stable synchronization equilibrium and analyze the corresponding basins of attraction.
\item We show that the network topology impacts synchronization by demonstrating that agents with similar neighborhood structures asymptotically synchronize independently of the signal function.
\end{enumerate}

The rest of the paper is organized as follows. In Section \ref{sec:problem_formulation}, we formally define the problem and introduce the proposed nonlinear consensus model. Section \ref{sec:analysis} provides a detailed equilibrium and stability analysis. In subsection \ref{subsec:invariance_of_the_state_space} and \ref{subsec:invariance_of_the_synchronization_manifold}, we analyze the invariance of the state space and the synchronization manifold, respectively. In subsection \ref{subsec:synchronization_induced_by_the_communication_signal}, we analyze the synchronization equilibria of the dynamics induced by the communication signal, and in subsection \ref{subsec:synchronization_properties_induced_by_the_graph_structure}, we analyze the synchronization properties induced by the graph structure. We conclude the paper with some remarks and future work in Section \ref{sec:conclusions}.

\textbf{Notation} In the following, we will denote by $\R$ and $\R_{\geq 0}$ the set of real numbers and the set of non-negative real numbers, respectively. For a vector $\boldsymbol{x} \in \R^N$, we denote by $x_i$ the $i$-th component of $\boldsymbol{x}$. For a matrix $\boldsymbol{A} \in \R^{N\times N}$, we denote by $a_{ij}$ the element of $\boldsymbol{A}$ at the $i$-th row and $j$-th column. We denote by $\boldsymbol{e}_i$ the $i$-th vector of the canonical basis of $\R^N$. We denote by $\one$ the vector of $\R^N$ with all components equal to $1$. Moreover, we denote by $\mathrm{diag}(\boldsymbol{x}) \in \R^{N\times N}$ the diagonal matrix with diagonal elements given by the vector $\boldsymbol{x} \in \R^{N}$. For a function $f: \mathcal{X} \to \mathcal{X}$, we denote $\mathrm{Fix}(f) = \left\{ x \in \mathcal{X} \mid x =  f(x) \right\}$ the set of fixed points of $f$ in $\mathcal{X}$.

\section{PROBLEM FORMULATION} \label{sec:problem_formulation}
We consider the classical multi-agent framework in which $N$ individuals/agents belonging to the set $\Vcal = \Vset$, interact according to an \emph{undirected fixed graph} $\Gcal = \left(\Vcal, \Ecal\right)$. We denote by $\boldsymbol{A} \in \R^{N\times N}$ the symmetric \emph{adjacency matrix} associated with the graph, i.e., $a_{ij} = 1$ if $(i,j) \in \Ecal$ and $a_{ij} = 0$ otherwise. We denote, by $\boldsymbol{D} \in \R^{N\times N}$ the \emph{degree matrix}, i.e., $\boldsymbol{D} $ is the diagonal matrix with diagonal elements defined by $d_i = \sum_{j=1}^{N} a_{ij}$ for all $i\in \Vcal$. We also introduce $\boldsymbol{L} := \boldsymbol{D} - \boldsymbol{A}\in \R^{N\times N}$ the associated \emph{Laplacian matrix}. The \emph{neighborhood} of the agent $i$ is denoted by $\mathcal{N}_i$ and represents the set of agents that influence $i$ according to the graph $\Gcal$ (i.e., $j \in \mathcal{N}_i \Leftrightarrow (j, i) \in \Ecal$). By definition, the cardinality of $\mathcal{N}_i$ is $d_i$. Let us also recall that a \emph{path} in $\Gcal$ is a finite sequence of edges $(i_1,i_2),(i_2,i_3),\dots,(i_p,i_{p+1})$ such that $(i_k,i_{k+1})\in \Ecal$  for all $k\in \{1,\dots,p\}$. Two vertices $i$, $j\in \Vcal$ are \emph{connected} in $\Gcal$ if there exists a path in $\Gcal$ joining $i$ and $j$ (i.e., $i_1=i$ and $j_p=j$). The graph $\Gcal$ is \emph{connected} if any two vertices are connected.

The state of agent $i$ is denoted by $x_i \in \left[-1,1\right]$ and evolve according to the following dynamics
\begin{equation}\label{eq:dynamic_single_agent}
	\dot{x}_i = \frac{1}{d_i} \sum_{j = 1}^{N} a_{ij} s(x_j) - x_i,
\end{equation}
where $s : \left[-1,1\right] \to \left[-1,1\right]$ is a \textbf{common non-decreasing Lipschitz-continuous function}. 

The nonlinear dynamics \eqref{eq:dynamic_single_agent} captures the imperfect communication between agents via a broad class of signal functions. For example, common choices include the linear map $s(x)=x$ or the affine function $s(x)=ax+b$, which respectively model unbiased or biased estimations. One can also consider sigmoidal functions, such as $s(x) = \tanh(x)$, that are inspired by biological phenomena \cite{grayMultiagentDecisionMakingDynamics2018} or neuroscience \cite{bizyaevaNonlinearOpinionDynamics2023}.
Furthermore, \eqref{eq:dynamic_single_agent} can describe the behavior of quantized communication, where discontinuities are inherent. Specifically, discontinuous quantization schemes, as in \cite{Martins2008,ceragioliConsensusDisagreementRole2018b,chowdhuryContinuousOpinionsDiscrete2016}, can be approximated using a continuous piecewise affine signal function wherein different linear segments connect at discontinuities. The qualitative behavior of the agents' states remains unchanged since, as we will show later, synchronization occurs only at specific fixed points of the signal function. In practice, such a signal function $s(x)$ can be interpreted as the expectation of a discontinuous quantized signal, potentially perturbed by a small noise term of small amplitude. In any case, \eqref{eq:dynamic_single_agent} provides a robust and analytically tractable framework for analyzing a wide range of communication processes. This approach bridges the gap between idealized linear consensus and more realistic scenarios involving nonlinear or even discontinuous interactions.

In the following, we will consider the collective form of the dynamics \eqref{eq:dynamic_single_agent} of the agents given by 
\begin{equation}\label{eq:dynamic}
	\dot{\boldsymbol{x}} = \boldsymbol{D}^{-1} \! \boldsymbol{A} \boldsymbol{s}(\boldsymbol{x}) - \boldsymbol{x} := \boldsymbol{f}(\boldsymbol{x}),
\end{equation}
where the state is denoted by $\boldsymbol{x} = (x_1,\dots,x_N)^\top \in \left[-1,1\right]^N$ and $\boldsymbol{s}(\boldsymbol{x})$ the vector $(s(x_1),\dots,s(x_N))^\top$. 

\begin{remark}
	In the case where $s(x) = x$, the dynamics \eqref{eq:dynamic} becomes the classical normalized linear consensus dynamics. i.e., $\dot{\boldsymbol{x}} = \boldsymbol{D}^{-1} \! \boldsymbol{A} \boldsymbol{x} - \boldsymbol{x} = - \boldsymbol{D}^{-1} \boldsymbol{L} \boldsymbol{x}$.
\end{remark}

This choice of the normalized adjacency matrix $\boldsymbol{D}^{-1} \! \boldsymbol{A}$ rather than the standard adjacency matrix $\boldsymbol{A}$ ensures uniformity in the speed of the dynamics across the network. Furthermore, since $\boldsymbol{D}^{-1} \! \boldsymbol{A}$ is a stochastic matrix, the Perron-Frobenius theorem can be applied, leading to the following lemma:
\begin{lemma}[Perron-Frobenius]\label{lemma:perron_frobenius}
Let $\Gcal$ be a connected graph. Then, the normalized adjacency matrix $\boldsymbol{D}^{-1} \! \boldsymbol{A}$ has a simple eigenvalue $1$, and all other eigenvalues have modulus strictly less than $1$. Moreover, the vector $\one$ is the right eigenvector associated with eigenvalue $1$ of $\boldsymbol{D}^{-1} \! \boldsymbol{A}$.
\end{lemma}

Consequently, as provided in Lemma~\ref{lemma:perron_frobenius}, normalization ensures that the spectral radius of $\boldsymbol{D}^{-1} \! \boldsymbol{A}$ is at most $1$ for connected graphs. This property allows the dynamics to be forward invariant for the synchronization manifold, as seen in Section \ref{subsec:invariance_of_the_synchronization_manifold}. Moreover, the qualitative behavior of the system remains consistent with the non-normalized case, as in \cite{fontanRoleFrustrationCollective2022}. Although normalization affects the timescales of the interactions by standardizing the influence of different nodes, it does not alter the fundamental structure of synchronization equilibria. 

\section{Analysis} \label{sec:analysis}

In this section, we analyze the proposed nonlinear consensus dynamics. We first establish invariance properties, proving that the state space and synchronization manifold remain forward invariant. Next, we characterize synchronization equilibria as the fixed points of the signal function and introduce underestimation and overestimation, which determine stability and convergence. We then derive stability conditions and identify attraction basins. Finally, we examine how network topology influences synchronization, showing that agents with symmetric neighborhoods, as well as those in all-to-all and bipartite graphs, asymptotically synchronize.

\subsection{Invariance of the state space} \label{subsec:invariance_of_the_state_space}
First, we will present the first invariant set, namely the set of admissible states for the dynamics to ensure that the state of the system remains in a bounded domain. 

\begin{proposition}
	Let $\Gcal$ be a connected graph. The set $\mathcal{X} = \left[-1,1\right]^N$ is forward invariant for the dynamics \eqref{eq:dynamic}. i.e., if $\boldsymbol{x}(0) \in \mathcal{X}$, then $\boldsymbol{x}(t) \in \mathcal{X}$ for all $t \geq 0$.
\end{proposition}

\begin{proof}
	Let $\boldsymbol{x} \in \mathcal{X}$, then for all $i \in \Vcal$, $-1 \leq x_i \leq 1$. Then, for all $i \in \Vcal$, $-1 \leq s(x_i) \leq 1$ by definition of $s$. Since $\boldsymbol{D}^{-1}\! \boldsymbol{A}$ is row stochastic, each row sums to 1. Then, for all $i \in \Vcal$, $-1 \leq\boldsymbol{e}_i^\top \boldsymbol{D}^{-1} \! \boldsymbol{A} \boldsymbol{s}(\boldsymbol{x}) \leq 1$. Since the vector field, \eqref{eq:dynamic} is Lipschitz-continuous, we can apply Nagumo's theorem \cite{blanchiniSetInvarianceControl1999} to ensure that the solution remains in $\mathcal{X}$. Then, an analysis at the set's boundary ensures that the solution remains inside of $\mathcal{X}$. Then, for any $i \in \Vcal$ such that $x_i = 1$, one has $\dot{x}_i = \boldsymbol{e}_i^\top \boldsymbol{D}^{-1} \! \boldsymbol{A} \boldsymbol{s}(\boldsymbol{x}) -1 \leq 0$. Similarly, for any $i \in \Vcal$ such that $x_i = -1$, one has $\dot{x}_i = \boldsymbol{e}_i^\top \boldsymbol{D}^{-1} \! \boldsymbol{A} \boldsymbol{s}(\boldsymbol{x}) + 1 \geq 0$. Thus, for all $i \in \Vcal$, $-1 \leq x_i(t) \leq 1$ for all $t \geq 0$.
\end{proof}

This proposition ensures that the dynamics \eqref{eq:dynamic} is well-defined on $\mathcal{X}$ and that the state of the system will always be in $\mathcal{X}$ justifying the choice of the set $\mathcal{X}$ as the state space of the system.

\subsection{Invariance of the synchronization manifold} \label{subsec:invariance_of_the_synchronization_manifold}

Let us now define the set of synchronization. This set is the set of states where all agents have the same state. 

\begin{definition}
	A state is a \emph{synchronization} if $x_i = x_j$ for all $i,j \in \Vcal$. The \emph{synchronization manifold} is defined as 
	\begin{equation*}
		\mathcal{S} = \left\{ \boldsymbol{x} \in \left[-1,1\right]^N \mid \forall i,j \in \Vcal, \, \, x_i = x_j \right\}.
	\end{equation*}

    For $c \in \left[-1,1\right]$, we denote by $\mathcal{S}_c = \left\{ \boldsymbol{x} \in \mathcal{S} \mid \boldsymbol{x} = c \one \right\}$ the synchronization at $c$ and for $M \subset \left[-1,1\right]$, by $\mathcal{S}_{M} = \left\{ \boldsymbol{x} \in \mathcal{S} \mid \exists p \in M, \, \boldsymbol{x} = p \one \right\}$ the synchronization manifold intersecting $M^N$.
\end{definition}

We will now establish the forward invariance of $\mathcal{S}$ for the dynamics \eqref{eq:dynamic}.
\begin{lemma}\label{lemma:forward_invariant}
	The synchronization manifold $\mathcal{S}$ is forward invariant for the dynamics \eqref{eq:dynamic}. i.e., if $\boldsymbol{x}(0) \in \mathcal{S}$, then $\boldsymbol{x}(t) \in \mathcal{S}$ for all $t \geq 0$.
\end{lemma}

\begin{proof}
	Let $\boldsymbol{x} \in \mathcal{S}$, i.e., $\boldsymbol{x} = c \one$ for some $c \in \left[-1,1\right]$. Then, for all $i,j \in \mathcal{V}$, one has that
	\begin{align*}
		\dot{x}_i - \dot{x}_j &= \left( \boldsymbol{e}_i - \boldsymbol{e}_j \right)^\top \! \left( \boldsymbol{D}^{-1}\!\boldsymbol{A}\boldsymbol{s}(\boldsymbol{x}) - \boldsymbol{x} \right)\\ 
		&= \left( \boldsymbol{e}_i - \boldsymbol{e}_j \right)^\top \! \left( \boldsymbol{D}^{-1}\!\boldsymbol{A}\boldsymbol{s}(c\one) - c\one \right)\\
		&= \left( s(c) - c \right) \left( \boldsymbol{e}_i - \boldsymbol{e}_j \right)^\top \! \one = 0.
	\end{align*}
	Since $\one$ is the eigenvector associated with eigenvalue $1$ of $\boldsymbol{D}^{-1}\boldsymbol{A}$, from Lemma~\ref{lemma:perron_frobenius}. Thus, the synchronization manifold is forward invariant.
\end{proof}

The previous lemma shows that once the agents are synchronized, they will remain synchronized for all times. This is a direct consequence of the fact that the weight matrix $\boldsymbol{D}^{-1} \boldsymbol{A}$ is a stochastic matrix. This property is not satisfied for the classical adjacency matrix $\boldsymbol{A}$ as it is not row stochastic in general.

\subsection{Synchronization induced by the communication signal} \label{subsec:synchronization_induced_by_the_communication_signal}

In this section, we analyze the synchronization equilibria of dynamics \eqref{eq:dynamic}. Whereas classical linear consensus dynamics admit any synchronization state as an equilibrium, we show that this property does not hold in our case. Instead, we prove that synchronization equilibria correspond precisely to the fixed points of the communication function $s$.
We will present stability results for the synchronization equilibria independently of the graph topology. 

\begin{proposition}\label{prop:synchronization_equilibria_are_fixed_points}
	Let $\Gcal$ be a connected graph. Then, the only synchronization equilibria of the dynamics \eqref{eq:dynamic} are the fixed points of $s$. i.e., $\mathcal{S}_{\mathrm{Fix}(s)}$ is the set of synchronization equilibria.
\end{proposition}

\begin{proof}
	Let $\boldsymbol{x}^* \in \mathcal{S}$ be a synchronization equilibrium. i.e., $\boldsymbol{x}^* = c \one$ for some $c \in \left[-1,1\right]$. Then, by Lemma~\ref{lemma:perron_frobenius}, one has
    \begin{equation*}
        \boldsymbol{x}^* = \boldsymbol{D}^{-1}\! \boldsymbol{A} \boldsymbol{s}(\boldsymbol{x}^*) = \boldsymbol{D}^{-1}\! \boldsymbol{A} \boldsymbol{s}(c \one) = s(c)\boldsymbol{D}^{-1}\! \boldsymbol{A}  \one = s(c) \one.
    \end{equation*}
    Since $\boldsymbol{x}^* = c \one$, one has $s(c) = c$. Thus, $c \in \mathrm{Fix}(s)$ and $\boldsymbol{x}^* \in \mathcal{S}_{\mathrm{Fix}(s)} $.
\end{proof}

The previous proposition establishes a complete characterization of the synchronization equilibria: they are in one-to-one correspondence with the fixed points of $s$. This is a fundamental difference from classical linear consensus dynamics, where any synchronization state is an equilibrium. Indeed, this property only holds in our case when $s(x) = x$, for which $\mathcal{S}_{\mathrm{Fix}(s)} = \mathcal{S}$. The synchronization equilibria form a strict subset of the synchronization manifold for any other communication function.

Before analyzing the stability of synchronization equilibria, let us introduce the notion of underestimation and overestimation for the communication function $s$.

\begin{definition}\label{def:underestimation}
    A signal function $s$ is said to be:
    \begin{itemize}
		\item an \emph{underestimation} at $x \in \left[-1,1\right]$ if $x(s(x) - x) \leq 0$. 
        \item an \emph{overestimation} at $x \in \left[-1,1\right]$ if $x(s(x) - x) \geq 0$.
        \item a \emph{perfect estimation} in interval $I$ if for all $x \in I$, $s(x) = x$. i.e., $s$ is an underestimation and overestimation in $I$.
        \item a \emph{consistent estimation} around $c \in \mathrm{Fix}(s)$ if there exists a neighborhood $I$ of $c$ where for all $x\in I$, $(x-c)(s(x) - x) \leq 0$.
    \end{itemize}
\end{definition}

In other words, a function is an underestimation (resp. overestimation) at a point if the signal it sends is closer (resp. further) to $0$ than the point itself. It is a local consistent estimation if it sends the signal closer to $c$ than the point itself in a neighborhood of $c$ (e.g., $s$ is a consistent estimation around $0$ if $s$ is an underestimation around $0$). All those properties can be extended to an interval $I$ if the property holds for all $x \in I$. When $I$ is the whole interval $[-1,1]$, we say that the property is global.

The following proposition analyzes the case where the function $s$ is an underestimation at all points in the interval $[-1,1]$, i.e., $s$ is globally underestimating.

\begin{proposition}\label{prop:stability_synchronization_for_underestimation}
	Let $\Gcal$ be a connected graph and $s$ be a globally underestimating signal function. Then, every solution of \eqref{eq:dynamic} starting in $\mathcal{X}$ approaches a synchronization equilibrium at a fixed point of $s$. i.e., $\lim_{t \to \infty} \boldsymbol{x}(t) = \boldsymbol{x}^* \in \mathcal{S}_{\mathrm{Fix}(s)}$ for all $\boldsymbol{x}(0) \in \mathcal{X}$.
\end{proposition}
\begin{proof}
	Consider the Lyapunov candidate function $\boldsymbol{V}: \mathcal{X} \to \R_{\geq 0} $ defined by
	\begin{equation*}
		\boldsymbol{V}(\boldsymbol{x}) = \frac{1}{2} \boldsymbol{x}^\top \boldsymbol{D} \boldsymbol{x},
	\end{equation*}
	where the degree matrix $\boldsymbol{D}$ is symmetric positive definite ($\boldsymbol{D} = \boldsymbol{D}^\top \succ 0$) since the graph $\Gcal$ is connected.
	
	Calculating the time derivative of $\boldsymbol{V}$ along the trajectories of system \eqref{eq:dynamic}, yields
	\begin{align*}
		\dot{\boldsymbol{V}}(\boldsymbol{x}) &= \boldsymbol{x}^\top \boldsymbol{D} \dot{\boldsymbol{x}} = \boldsymbol{x}^\top \boldsymbol{D}(\boldsymbol{D}^{-1} \!\boldsymbol{A} \boldsymbol{s}(\boldsymbol{x}) - \boldsymbol{x}) \\
		&= \boldsymbol{x}^\top \boldsymbol{A} \boldsymbol{s}(\boldsymbol{x}) - \boldsymbol{x}^\top \boldsymbol{D} \boldsymbol{s}(\boldsymbol{x}) + \boldsymbol{x}^\top \boldsymbol{D} (\boldsymbol{s}(\boldsymbol{x}) - \boldsymbol{x}) \\
		&= - \boldsymbol{x}^\top \boldsymbol{L} \boldsymbol{s}(\boldsymbol{x}) + \boldsymbol{x}^\top \boldsymbol{D} (\boldsymbol{s}(\boldsymbol{x}) - \boldsymbol{x}),
	\end{align*}
	where $\boldsymbol{L} = \boldsymbol{D} - \boldsymbol{A}$ is the Laplacian matrix.

Let us demonstrate that both terms in $\dot{\boldsymbol{V}}(\boldsymbol{x})$ are non-positive. For the first term, since $s$ is non-decreasing, we have $(x_i - x_j)(s(x_i) - s(x_j)) \geq 0$ for all $i,j \in \Vcal$. This implies:
\begin{align*}
    \boldsymbol{x}^\top \boldsymbol{L} \boldsymbol{s}(\boldsymbol{x}) &= \sum_{1 \leq i, j \leq N} a_{ij} x_i \left(s(x_i) - s(x_j)\right) \\
    &= \frac{1}{2} \sum_{1 \leq i, j \leq N} a_{ij} \left(x_i - x_j\right)\left(s(x_i) - s(x_j)\right),
\end{align*}
where the second equality follows from the symmetry of the adjacency matrix ($a_{ij} = a_{ji} \geq 0$) due to $\Gcal$ being undirected. Therefore, we have that $-\boldsymbol{x}^\top \boldsymbol{L} \boldsymbol{s}(\boldsymbol{x}) \leq 0$.

For the second term, the global underestimation property of $s$ ensures that $x_i(s(x_i) - x_i) \leq 0$ for all $i \in \Vcal$ and $x_i \in [-1,1]$. Consequently:
\begin{equation*}
    \boldsymbol{x}^\top \boldsymbol{D} \left(\boldsymbol{s}(\boldsymbol{x}) - \boldsymbol{x}\right) = \sum_{1 \leq i \leq N} d_i x_i(s(x_i) - x_i) \leq 0,
\end{equation*}
where the non-negativity of the degrees $d_i$ preserves the inequality.

Let us now characterize the largest invariant set where $\dot{\boldsymbol{V}}(\boldsymbol{x}) = 0$. The condition $\dot{\boldsymbol{V}}(\boldsymbol{x}) = 0$ is satisfied if and only if both terms in the derivative vanish simultaneously. The first term $-\boldsymbol{x}^\top \boldsymbol{L} \boldsymbol{s}(\boldsymbol{x}) = 0$ implies that $x_i = x_j$ for all $i,j \in \Vcal$, while the second term $\boldsymbol{x}^\top \boldsymbol{D} \left(\boldsymbol{s}(\boldsymbol{x}) - \boldsymbol{x} \right) = 0$ implies that $x_i = s(x_i)$ for all $i \in \Vcal$. The first condition establishes that $\boldsymbol{x} \in \mathcal{S}$, which by Lemma~\ref{lemma:forward_invariant} is forward invariant. Therefore, the largest invariant set satisfying $\dot{\boldsymbol{V}}(\boldsymbol{x}) = 0$ can be characterized as $\mathcal{S} \cap \mathrm{Fix}(s)^N$ which is precisely the set $\mathcal{S}_{\mathrm{Fix}(s)}$.

The convergence of solutions follows from LaSalle's invariance principle \cite[Theorem 3.4]{khalil_nonlinear_2002}. Given that $\mathcal{X}$ is compact and forward invariant, $\boldsymbol{V}(\boldsymbol{x})$ is bounded from below on $\mathcal{X}$. Moreover, since $\dot{\boldsymbol{V}}(\boldsymbol{x}) \leq 0$, and $\mathcal{S}_{\mathrm{Fix}(s)}$ is the largest invariant set where $\dot{\boldsymbol{V}}(\boldsymbol{x}) = 0$, we can conclude that every solution with initial condition in $\mathcal{X}$ converges to $\mathcal{S}_{\mathrm{Fix}(s)}$ as $t \to \infty$.
\end{proof}

Then, in case of global underestimation, the convergence to a fixed point of $s$ is guaranteed. This behavior is illustrated in Figure \ref{fig:opinion_underestimation}. 


\begin{figure}
    \centering
    \vspace{0.1cm}
    \includegraphics[width=\linewidth]{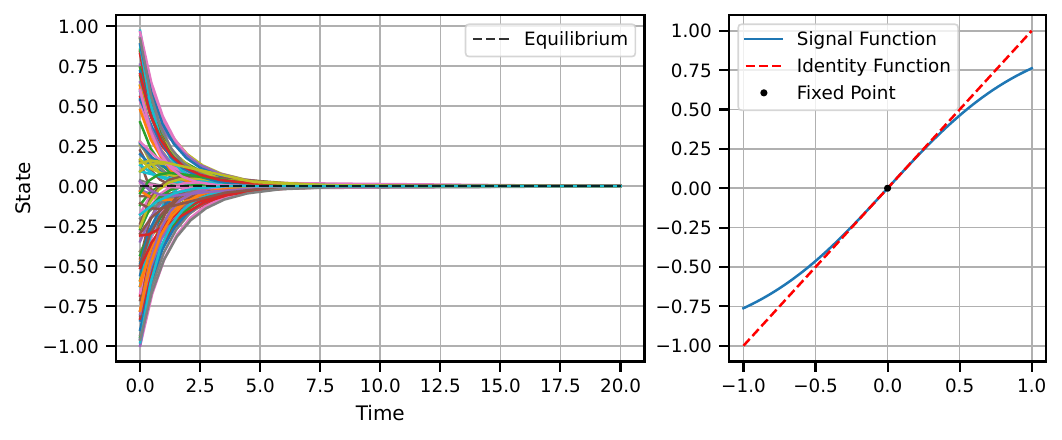}
    \caption{State evolution for agents with global underestimation signal over a connected graph (Erd\H{o}s-Rényi with $N=100$ and $p=0.1$) with random uniform initial conditions and $s(x) = \tanh(x)$.}
    \label{fig:opinion_underestimation}
\end{figure}
Having established the existence and convergence properties of synchronization equilibria, a fundamental question concerns their stability characteristics. The local behavior of the signal function $s$ around these equilibria plays a crucial role in determining their stability properties. The following theorem provides a complete characterization of the stability of synchronization equilibria based on the local estimation properties of $s$.

\begin{theorem}\label{thm:stability_consensus}
	Let $\Gcal$ be a connected graph and $\mathcal{S}_c=c \one$ be a synchronization equilibrium. Then the following holds true:
	\begin{enumerate}
		\item $\mathcal{S}_c$ is locally stable if and only if $s$ is a consistent estimation around $c$. \label{itm:1}
		\item $\mathcal{S}_c$ is locally asymptotically stable if and only if $\mathcal{S}_c$ is locally stable and $c$ is an isolated fixed point of $s$. \label{itm:2}
	\end{enumerate}
\end{theorem}

\begin{proof}
\ref{itm:1}) Let us start by showing that if $\mathcal{S}_c$ is locally stable, then $s$ must be a consistent estimation around $c$. Indeed, if $\mathcal{S}_c$ is locally stable, there exists a forward invariant neighborhood $\boldsymbol{\Omega}$ of $\mathcal{S}_c$ for the dynamics \eqref{eq:dynamic}. Let us note $I = \boldsymbol{\Omega} \cap \mathcal{S} \subset \left[-1,1\right]$. By Proposition~\ref{lemma:forward_invariant}, $\mathcal{S}_{I}$ is also forward invariant for \eqref{eq:dynamic}. Then, by considering the dynamics along the synchronization manifold, one has for all $x \in I$, 
	\begin{equation*}
		\sign{\dot{x}} = - \sign{s(x) - x} \text{ or } \dot{x} = s(x) - x = 0.
	\end{equation*}
    where $\sign{x}$ is the sign of $x$. This yields that for all $x \in I$, one has $x < c$ implies $s(x) \leq x$ and $x > c$ implies $s(x) \geq x$. Combining the two gives $(x-c)(s(x)-x) \leq 0$, which proves that $s$ must be a consistent estimation around $c$.

	Conversely, let us show that if $s$ is a consistent estimation around $c$, then $\mathcal{S}_c$ is locally stable. Let $s$ be a consistent estimation around $c$. Then, by definition, there exists a neighborhood $I$ of $c$ such that $s$ is an overestimation for $x < c$ and an underestimation for $x > c$. Using a similar Lyapunov function as in the proof of Proposition~\ref{prop:stability_synchronization_for_underestimation}, we can show that $\mathcal{S}_c$ is locally stable. Indeed, the Lyapunov function is now given by $\boldsymbol{V}(\boldsymbol{x}) = \left(\boldsymbol{x} - c \one\right)^\top \! \boldsymbol{D} \left(\boldsymbol{x} - c \one\right)/2$. The time derivative of $\boldsymbol{V}$ along the trajectories of system \eqref{eq:dynamic} is given by
	\begin{equation*}
		\dot{\boldsymbol{V}}(\boldsymbol{x}) = - \left(\boldsymbol{x} - c \one\right)^\top \! \boldsymbol{L} \boldsymbol{s}(\boldsymbol{x}) + \left(\boldsymbol{x} - c \one\right)^\top \! \boldsymbol{D} (\boldsymbol{s}(\boldsymbol{x}) - \boldsymbol{x}),
	\end{equation*}
	which is non-positive since $s$ is a consistent estimation around $c$. Then, $\mathcal{S}_c$ is locally stable by \cite[Theorem 3.1]{khalil_nonlinear_2002}.

	\ref{itm:2}) If in addition, $c$ is an isolated fixed point, then $c\one$ is locally attractive, since in a neighborhood $\boldsymbol{\Omega}$ of $\mathcal{S}_c$  one has for all $\boldsymbol{x} \in \boldsymbol{\Omega}$, $\dot{\boldsymbol{V}}(\boldsymbol{x})\leq 0$ and $\dot{\boldsymbol{V}}(\boldsymbol{x}) = 0$ if and only if $\boldsymbol{x} = c \one$. Thus, $\mathcal{S}$ is locally asymptotically stable.
    
    Conversely, if $\mathcal{S}_c$ is locally attractive, it must be attractive also along the synchronization manifold. Then, for the one dimensional dynamics on $\mathcal{S}$, there exists a neighborhood $I \subset \left[ -1,1\right]$ of $c$ such that for all $x\in I$, $\sign{\dot{x}} = - \sign{s(x) - x}$ and $\dot{x} = 0$ if and only if $x = c$. This yields that $c$ is the only fixed point of $s$ contained in $I$, meaning it is isolated.
\end{proof}

The theorem characterizes the stability conditions for synchronization equilibria in terms of the behavior of the function $s$ near fixed points. Specifically, a synchronization equilibrium $\mathcal{S}_c$ is locally asymptotically stable (AS) if and only if the function satisfies specific local estimation properties around $c$: for points $x < c$ in a neighborhood of $c$, the function $s$ must overestimate, while for points $x > c$, the function $s$ must underestimate.  

\begin{remark}\label{remark:quantization_function_stability}
    While Lipschitz continuity of $s$ was not explicitly used in the proof, it is assumed to ensure unique solutions and avoid technical complications. Nevertheless, the same result can be shown for approximations of non-Lipschitz-continuous functions as long as they are non-decreasing. For instance, consider the quantization function $q(x) = \mathrm{sign}(x)$ from \cite{chowdhuryContinuousOpinionsDiscrete2016} and \cite{prisantDisagreementLimitCycles2024}. For any $\varepsilon >0$, one may approximate $q(x)$ by $s_{\varepsilon}(x) = q(x)$ for  $x \in [-1,1]  \setminus \{[-\varepsilon, \varepsilon]\}$ and $s_{\varepsilon}(x) = x/\varepsilon$ for $x \in [-\varepsilon, \varepsilon]$. Then, by Theorem~\ref{thm:stability_consensus}, the same result holds for any $\varepsilon > 0$. i.e., the only asymptotically stable synchronization equilibria are $\left\{-\one, \one\right\}$.
\end{remark} 

This result provides a rigorous framework for analyzing the local stability of synchronization equilibria at fixed points of $s$ for connected graphs for a general signal function $s$. Figure \ref{fig:consensus_stability} illustrates this phenomenon, contrasting a stable consensus equilibrium (left panel) with an unstable consensus equilibrium (right panel).


\begin{figure}
	\centering
    \vspace{0.3cm}
	\begin{subfigure}[t]{0.45\linewidth}
		\centering
		\tikzset{every picture/.style={line width=0.75pt}} 

		\begin{tikzpicture}[x=0.75pt,y=0.75pt,yscale=-1,xscale=1]

			\draw [color={rgb, 255:red, 208; green, 2; blue, 27 }  ,draw opacity=1 ] [dash pattern={on 4.5pt off 4.5pt}]  (249.83,80.17) -- (181.25,148.75) -- (170.33,159.67) ;
			\draw [line width=1.5]    (169.57,140.5) .. controls (184.24,125.2) and (192,122.21) .. (210.08,119.92) .. controls (228.17,117.62) and (247.14,99.64) .. (250.29,90.36) ;
			\draw  [color={rgb, 255:red, 1; green, 113; blue, 241 }  ,draw opacity=1 ][fill={rgb, 255:red, 0; green, 124; blue, 255 }  ,fill opacity=1 ] (209.7,121.72) .. controls (208.7,121.51) and (208.07,120.53) .. (208.28,119.53) .. controls (208.49,118.54) and (209.47,117.9) .. (210.47,118.11) .. controls (211.46,118.32) and (212.1,119.3) .. (211.89,120.3) .. controls (211.68,121.3) and (210.7,121.93) .. (209.7,121.72) -- cycle ;
			\draw    (170.15,100.08) -- (198,100.08) ;
			\draw [shift={(200,100.08)}, rotate = 180] [color={rgb, 255:red, 0; green, 0; blue, 0 }  ][line width=0.75]    (10.93,-3.29) .. controls (6.95,-1.4) and (3.31,-0.3) .. (0,0) .. controls (3.31,0.3) and (6.95,1.4) .. (10.93,3.29)   ;
			\draw    (250.14,140.08) -- (222.45,140.08) ;
			\draw [shift={(220.45,140.08)}, rotate = 360] [color={rgb, 255:red, 0; green, 0; blue, 0 }  ][line width=0.75]    (10.93,-3.29) .. controls (6.95,-1.4) and (3.31,-0.3) .. (0,0) .. controls (3.31,0.3) and (6.95,1.4) .. (10.93,3.29)   ;

			\draw (221.07,115.89) node [anchor=north west][inner sep=0.75pt]    {$s( x)$};
			\draw (216.9,82.71) node [anchor=north west][inner sep=0.75pt]  [color={rgb, 255:red, 208; green, 2; blue, 27 }  ,opacity=1 ]  {$x$};
			\draw (198.32,102.36) node [anchor=north west][inner sep=0.75pt]    {$\textcolor[rgb]{0,0.46,1}{c}$};

		\end{tikzpicture}

			\subcaption{AS: consistent estimation around $c$.}
	\end{subfigure}
	\hfil
	\begin{subfigure}[t]{0.45\linewidth}
		
		\centering

\tikzset{every picture/.style={line width=0.75pt}} 

\tikzset{every picture/.style={line width=0.75pt}} 

\begin{tikzpicture}[x=0.75pt,y=0.75pt,yscale=-1,xscale=1]

\draw [color={rgb, 255:red, 208; green, 2; blue, 27 }  ,draw opacity=1 ] [dash pattern={on 4.5pt off 4.5pt}]  (229.83,60.17) -- (150.33,139.67) ;
\draw [line width=1.5]    (170.62,140.23) .. controls (185.48,123.18) and (188.14,121.18) .. (190.08,99.92) .. controls (190.58,94.42) and (191.59,89.78) .. (192.96,85.76) .. controls (196.89,74.22) and (203.81,67.76) .. (210.31,60.69) ;
\draw  [color={rgb, 255:red, 0; green, 117; blue, 255 }  ,draw opacity=1 ][fill={rgb, 255:red, 0; green, 118; blue, 255 }  ,fill opacity=1 ] (189.7,101.72) .. controls (188.7,101.51) and (188.07,100.53) .. (188.28,99.53) .. controls (188.49,98.54) and (189.47,97.9) .. (190.47,98.11) .. controls (191.46,98.32) and (192.1,99.3) .. (191.89,100.3) .. controls (191.68,101.3) and (190.7,101.93) .. (189.7,101.72) -- cycle ;
\draw    (200.15,110.08) -- (228,110.08) ;
\draw [shift={(230,110.08)}, rotate = 180] [color={rgb, 255:red, 0; green, 0; blue, 0 }  ][line width=0.75]    (10.93,-3.29) .. controls (6.95,-1.4) and (3.31,-0.3) .. (0,0) .. controls (3.31,0.3) and (6.95,1.4) .. (10.93,3.29)   ;
\draw    (180.31,89.92) -- (152.62,89.92) ;
\draw [shift={(150.62,89.92)}, rotate = 360] [color={rgb, 255:red, 0; green, 0; blue, 0 }  ][line width=0.75]    (10.93,-3.29) .. controls (6.95,-1.4) and (3.31,-0.3) .. (0,0) .. controls (3.31,0.3) and (6.95,1.4) .. (10.93,3.29)   ;

\draw (165.02,61.99) node [anchor=north west][inner sep=0.75pt]    {$s( x)$};
\draw (215.18,77.85) node [anchor=north west][inner sep=0.75pt]  [color={rgb, 255:red, 208; green, 2; blue, 27 }  ,opacity=1 ]  {$x$};
\draw (176.35,97.05) node [anchor=north west][inner sep=0.75pt]    {$\textcolor[rgb]{0,0.46,1}{c}$};

\end{tikzpicture}

\subcaption{Unstable: not consistent estimation around $c$.}
	\end{subfigure}
	\caption{Illustration of the local stability conditions for synchronization equilibria as described in Theorem \ref{thm:stability_consensus}.}
	\label{fig:consensus_stability}
\end{figure}
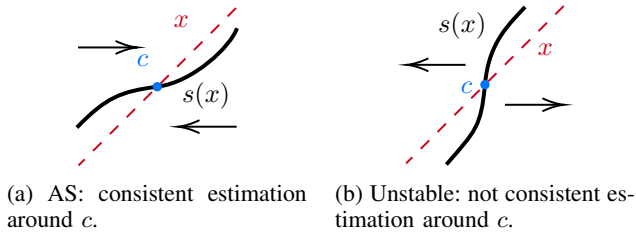

One may ask what the domain of attraction for a synchronization equilibrium is. With this aim, let us consider the following proposition.
\begin{proposition}\label{prop:attraction_domain_synchronization_equilibrium}
	Let $k_1 < k_2 < \dots < k_m$ be the ordered fixed points around which $s$ is inconsistent. 
    Then, for all $l \in \{1, \dots, m-1\}$, the sets
	\begin{equation*}
		\mathcal{K}_l = \left[k_l,k_{l+1}\right]^N \setminus \{k_l \one, k_{l+1}\one\}
	\end{equation*}
	are forward invariant for the dynamics \eqref{eq:dynamic}. 
	
	Additionally, let $I = \left(k_l,k_{l+1}\right)$ and denote $C = \mathrm{Fix}(s) \cap I$. Then, the attraction domain of $\mathcal{S}_C$ contains at least $\mathcal{K}_l$.
\end{proposition}

\begin{proof}
	To show the forward invariance of $\mathcal{K}_l$,
    as before, it is sufficient to prove that the dynamics \eqref{eq:dynamic} is pointing inward on the boundary of $\mathcal{K}_l$.

	Let $x \in \left\{k_l,k_{l+1}\right\}^N \setminus \{k_l \one, k_{l+1}\one\}$, and denote, for an node $i$, $n_i^-$ and $n_i^+$ the number of neighbors with state $k_{l}$ and $k_{l+1}$, respectively. Then, for all $i \in \{1, \dots, N\}$ one has,
	\begin{align*}
		\dot{x}_i &= \frac{1}{d_i} \sum_{j=1}^N a_{ij} s(x_j) - x_i = \frac{1}{d_i} \left( n_i^+ k_{l+1} - n_i^- k_l \right) - x_i,
	\end{align*}
	where $k_l < (n_i^+ k_{l+1} + n_i^- k_{l})/d_i < k_{l+1}$ since $n_i^+ + n_i^- = d_i$ and $n_i^+, n_i^- > 0$. Thus, if $x_i = k_l$, then $\dot{x}_i > 0$ and if $x_i = k_{l+1}$, then $\dot{x}_i < 0$. Then, $\mathcal{K}_l$ is forward invariant, proving the first part of the proposition.

	Now, let us show the existence of a fixed point of $s$ in $I$ such that $\mathcal{S}_C$ is forward invariant. The function $g(x) = s(x) - x$ is continuous on $I$, for small $\varepsilon >0$ one has $g(k_l + \varepsilon) > 0$ and $g(k_{l+1} - \varepsilon) < 0$. Then, by the intermediate value theorem, there exists $c \in \left(k_l,k_{l+1}\right)$ such that $g(c) = 0$. Then, $C$ is nonempty. Moreover, $s$ must be a consistent estimation around any $c \in C$ by definition of $k_1, k_2,\dots, k_m$. Then, by Theorem~\ref{thm:stability_consensus}, for all $c \in C$, $c\one$ is a stable equilibrium of \eqref{eq:dynamic} and $\mathcal{S}_C$ is forward invariant for \eqref{eq:dynamic}. Finally, by using $\boldsymbol{V}(\boldsymbol{x}) = \left(\boldsymbol{x} - c \one\right)^\top \! \boldsymbol{D}\left(\boldsymbol{x} - c\one\right)/2$, and applying the LaSalle's invariance principle \cite[Theorem 3.4]{khalil_nonlinear_2002}, one has that the attraction domain of $\mathcal{S}_C$ contains at least $\mathcal{K}_l$.
\end{proof}

When $s$ is a globally overestimating function, the following corollary holds true.

\begin{figure}
    \centering
    \vspace{0.1cm}
    \includegraphics[width=\linewidth]{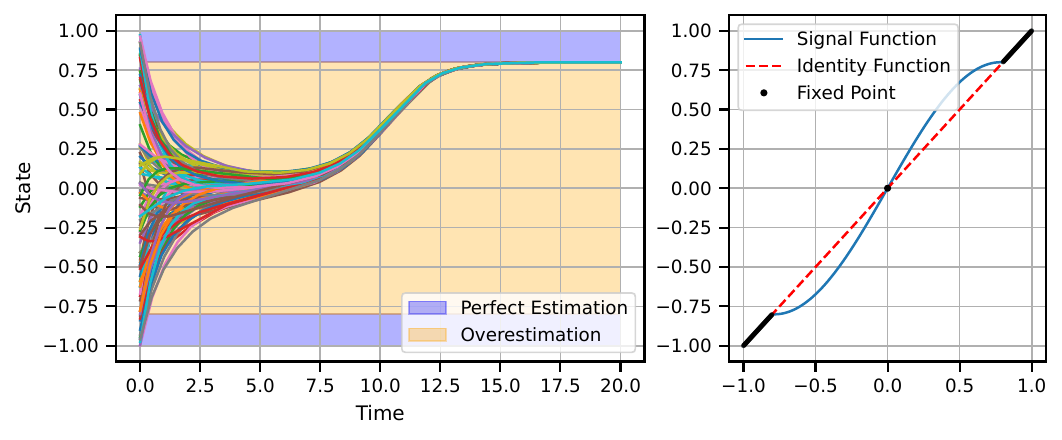}
    \caption{State evolution for agents with a globally overestimating signal under the same conditions as in Figure \ref{fig:opinion_underestimation}, $s(x) = 0.8 \sin(2x)/\sin(1.6)$ for $x\in [-0.8,0.8]$ and $s(x) = x$ otherwise.}
    \label{fig:opinion_overestimation}
\end{figure}

\begin{corollary}
    Let $s$ be a globally overestimating function and let $x(0) \in \left[0,1\right]^N$ or $x(0) \in \left[-1,0\right]^N$. Then, the dynamics \eqref{eq:dynamic} converges to the synchronization equilibrium $\mathcal{S}_{\mathrm{Fix}(s)}$.
\end{corollary}

The proof of this corollary is straightforward due to the Proposition~\ref{prop:attraction_domain_synchronization_equilibrium}. Indeed, if $s$ is globally overestimating, then the function $s$ is overestimating for all $x \in \left[0,1\right]$ or $x \in \left[-1,0\right]$ with $-1$ and $1$ as fixed points. Then, by Proposition~\ref{prop:attraction_domain_synchronization_equilibrium}, the dynamics \eqref{eq:dynamic} converges to one of the synchronization equilibrium in $\mathcal{S}_{\mathrm{Fix}(s)}$. However, the attraction basin of $\mathcal{S}_{\mathrm{Fix}(s)}$ is not restricted to $\left[0,1\right]^N$ or $\left[-1,0\right]^N$ as illustrated in Figure \ref{fig:opinion_overestimation}. In this figure, we observe a contraction of the states of the agents around an unstable equilibrium due to the topological properties of the graph, and then states enter in 
$\left[0,1\right]^N$ and converge to a synchronization equilibria in $\mathcal{S}_{\mathrm{Fix}(s)}$. This motivates the analysis of graph structure on the convergence properties of the dynamics \eqref{eq:dynamic}.

\begin{remark}
	In the same spirit as in Remark~\ref{remark:quantization_function_stability}, analogous attraction domain results can be derived for the discontinuous quantization function $q(x) = \mathrm{sign}(x)$ from \cite{chowdhuryContinuousOpinionsDiscrete2016} and \cite{prisantDisagreementLimitCycles2024}. In this case, the attraction domain of the synchronization manifold contains the same hypercubes as a globally overestimating function. i.e., $\left[0,1\right]^N$ and $\left[-1,0\right]^N$.
\end{remark}

\subsection{Synchronization induced by the graph structure} \label{subsec:synchronization_properties_induced_by_the_graph_structure}

In this section, we study the synchronization equilibria of the dynamics \eqref{eq:dynamic} for the case of symmetric neighbors. We will show how the symmetry of the neighbors can influence the synchronization equilibria. 


\begin{proposition}\label{prop:symmetry_synchronization}
	Let $i$, $j\in \Vcal$ such that $i\neq j$ and either: ``$\mathcal{N}_i = \mathcal{N}_j$" or ``$i\in \mathcal{N}_j$, $j\in \mathcal{N}_i$ and $\mathcal{N}_j \setminus\left\{i\right\} = \mathcal{N}_i \setminus \left\{j\right\} $". Then, agents $i$ and $j$ asymptotically synchronize. i.e., $\lim_{t \to \infty} x_i(t) - x_j(t)=0$.
\end{proposition}

\begin{proof}
	Let $i$, $j\in \Vcal$ such that $i\neq j$ and either: ``$\mathcal{N}_i = \mathcal{N}_j$" or ``$i\in \mathcal{N}_j$, $j\in \mathcal{N}_i$ and $\mathcal{N}_j \setminus\left\{i\right\} = \mathcal{N}_i \setminus \left\{j\right\} $". We consider the dynamics of the difference $\delta_{ij} = x_i - x_j$. Then, we have
	\begin{align}
		\dot{\delta}_{ij} &= \dot{x}_i - \dot{x}_j = \left( \boldsymbol{e}_i - \boldsymbol{e}_j \right)^\top \! \left( \boldsymbol{D}^{-1} \! \boldsymbol{A} \boldsymbol{s}(\boldsymbol{x}) - \boldsymbol{x} \right) \nonumber\\
		&= \left( \boldsymbol{e}_i - \boldsymbol{e}_j \right)^\top \! \boldsymbol{D}^{-1}\! \boldsymbol{A} \boldsymbol{s}(\boldsymbol{x}) - \delta_{ij}. \label{eq:dynamics_delta}
	\end{align}
	Now, due to the symmetry of the neighborhood, we have that for all $\boldsymbol{y} \in \R^N$, $\left( \boldsymbol{e}_i - \boldsymbol{e}_j \right)^\top \boldsymbol{D}^{-1} \! \boldsymbol{A} \boldsymbol{y}  = d_{i}^{-1} a_{ij} y_j - d_{j}^{-1} a_{ji} y_i = d_{i}^{-1} a_{ij} \left(y_j - y_i\right)$. Then, we have
	\begin{align*}
		&\sign{\dot{\delta}_{ij}} = \sign{d_{i}^{-1} a_{ij} \left(  s(x_j) - s(x_i) \right) - \delta_{ij}} \\
		& = -\sign{ d_{i}^{-1}a_{ij} \left(s(x_i) - s(x_j)\right) + \delta_{ij}}  = -\sign{ \delta_{ij}},
	\end{align*}
	since the function $s$ is non-decreasing and $d_{i}^{-1} a_{ij} \geq 0$. Then, $\lim_{t \to \infty} \delta_{ij}(t) = 0$ and agents $i$ and $j$ asymptotically synchronize. 
\end{proof}

This result only relies on the neighborhoods of agents $i$ and $j$ and does not depend on global properties of the graph $\Gcal$. The following corollary is then immediate.
\begin{corollary}\label{coro:all_to_all}
	Let $\Gcal$ be the all-to-all graph or a complete bipartite graph. Then, all agents asymptotically synchronize.
\end{corollary}

\begin{remark}
	Results from Proposition~\ref{prop:symmetry_synchronization} and Corollary~\ref{coro:all_to_all} depend only on the monotonicity of the signal function $s$. 
    In case of a discontinuous quantization, the result is still valid as long as the function is non-decreasing on the set of the agents' states \cite{ceragioliConsensusDisagreementRole2018b}.
\end{remark}

It is noteworthy that convergence to a synchronization equilibrium is not universally guaranteed for every choice of signal function. In particular, when the signal function exhibits overestimating behavior, even connected graphs may fail to achieve synchronization if they are sparse. Figure \ref{fig:opinion_asynchronization} illustrates this phenomenon using a line graph with six agents, where each agent has at most two neighbors. Under these conditions, the overestimating nature of the signal may prevent the agents' states from converging to a common state.

\begin{figure}
	\centering
    \vspace{0.1cm}
	\includegraphics[width=\linewidth]{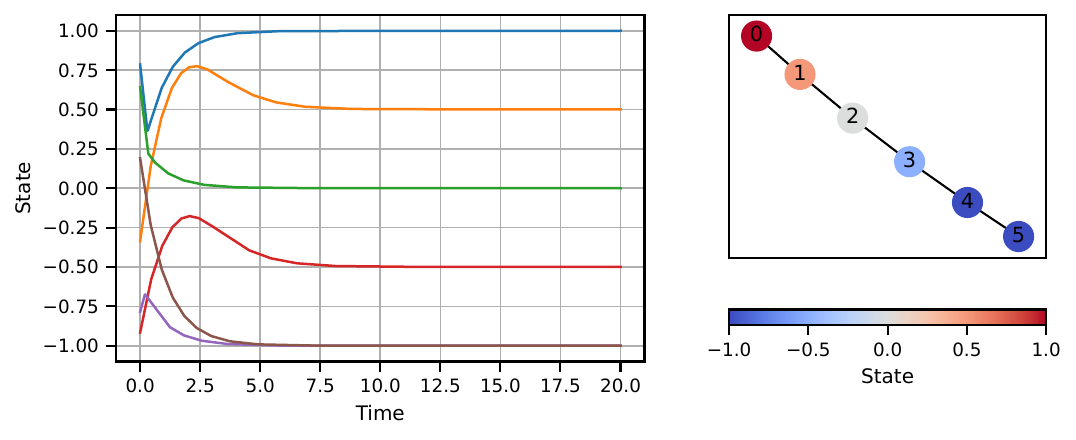}
	\caption{State evolution for agents with a globally overestimating signal ($s(x) = \tanh(20 x)$) over a line graph of 6 agents. The agents do not converge to a synchronization equilibrium for these random initial conditions.}
	\label{fig:opinion_asynchronization}
\end{figure}

\section{CONCLUSIONS AND FUTURE WORK}\label{sec:conclusions}

In this paper, we introduced a framework for continuous-time consensus-like dynamics that generalizes traditional linear consensus and allows a nonlinear Lipschitz-continuous signaling function as the communication medium. Our analysis rigorously characterized the synchronization equilibria as the fixed points of the communication signal and established both local and, in some particular cases, global stability conditions based solely on the properties of the signal function and the underlying network topology. 

These results provide insights on how communication nonlinearity and network connectivity jointly determine multi-agent coordination. They also open several promising research directions, such as: analyzing the asymptotic behavior under non-smooth interaction signals, studying the convergence properties when the signal is not globally underestimating the state, etc. 

\addtolength{\textheight}{-12cm}   








\bibliographystyle{IEEEtran}
\bibliography{COSA}

\end{document}